\documentclass[11pt,reqno]{amsart}
\usepackage{amssymb,latexsym}
\usepackage{graphicx}
\usepackage{url}		

\usepackage{multirow,bigdelim}
\usepackage{arydshln}

\calclayout

\makeatletter
\newcommand{\monthyear}[1]{%
  \def\@monthyear{\uppercase{#1}}}
\newcommand{\volnumber}[1]{%
  \def\@volnumber{\uppercase{#1}}}
\AtBeginDocument{%
\def\ps@plain{\ps@empty
  \def\@oddfoot{\@monthyear \hfil \thepage}%
  \def\@evenfoot{\thepage \hfil \@volnumber}}
\def\ps@firstpage{\ps@plain}
\def\ps@headings{\ps@empty
  \def\@evenhead{%
    \setTrue{runhead}%
    \def\thanks{\protect\thanks@warning}%
    \uppercase{}\hfil}
  \def\@oddhead{%
    \setTrue{runhead}%
    \def\thanks{\protect\thanks@warning}%
    \hfill\uppercase{Fibonacci-like sequences for variants of the Tower of Hanoi}}%
  \let\@mkboth\markboth
  \def\@evenfoot{%
    \thepage \hfil \@volnumber}%
  \def\@oddfoot{%
    \@monthyear \hfil \thepage}%
  }%
\footskip=25pt
\pagestyle{headings}%
}
\makeatother

\newcommand{\N}{{\mathbb N}}
\newcommand{\Z}{{\mathbb Z}}

\newcommand{\card}{{\rm card}}

\theoremstyle{plain}
\numberwithin{equation}{section}
\newtheorem{thm}{Theorem}[section]
\newtheorem{theorem}[thm]{Theorem}

\newtheorem{definition}[thm]{Definition}

\newtheorem{corollary}[thm]{Corollary}

\begin{document}
\monthyear{}
\volnumber{}
\setcounter{page}{1}

\title{Fibonacci-like Sequences for variants of the Tower of Hanoi, and corresponding graphs and Gray Codes}
\author{Beno\^{\i}t Rittaud}
\address{Laboratoire Analyse, Géométrie et Applications (CNRS UMR 7539)\\
                Institut Galil\'ee\\
                Universit\'e Sorbonne Paris Nord\\
                93430 Villetaneuse, France}
\email{rittaud@math.univ-paris13.fr}

\begin{abstract}
We modify the rules of the classical Tower of Hanoi puzzle in a quite natural way to get the Fibonacci sequence involved in the optimal algorithm of resolution, and show some nice properties of such a variant. In particular, we deduce from this {\em Tower of Hanoi-Fibonacci} a Gray-like code on the set of binary words without the factor $11$, which has some  properties intersting for itself and from which an iterative algorithm for the Tower of Hanoi-Fibonacci is obtained. Such an algorithm involves the Fibonacci substitution. Eventually, we briefly extend the study to some natural generalizations.\end{abstract}

\maketitle

The Tower of Hanoi is a puzzle invented by Édouard Lucas \cite{Lucas1884, Lucas1883} in which a set of $n$ disks of different radius from 1 to $n$ are put on a peg $A$ in decreasing order, thus materializing a triangular tower. Two other pegs, $B$ and $C$, are empty. The aim of the game is to move all the disks on the peg $C$ (or, in a roughly equivalent version, on either $B$ and $C$), following the two rules: (1) the disks are moved one at a time, taking a disk on the top of a peg and putting it on the top of another peg, and (2) a disk cannot be put over a smaller disk. This is what we will call the {\em classical} Tower of Hanoi puzzle. A set-theoretic description of it is the following. We write either $d_k$ or $k$ for the disk of radius $k$, and $\Delta_k=\{1,\ldots, k\}$ for the set of the $k$ smallest disks (with $\Delta_k=\varnothing$ for $k<1$). Any state of the puzzle corresponds to an ordered $3$-partition of $\Delta_n$, written as $(A,B,C)$ and referred as a {\em state} of the puzzle (also referred as {\em regular state} in the litterature, when it is necessary to emphasize on the fact that disks on each peg has to be set up with decreasing size, an assumption that is weakened in some studies). A move from such a state to another one, say $(A',B',C')$, is allowed if, and only if, the two ordered partitions are equal up to some $d\in\Delta_n$ such that $d\in\{\min(A),\min(B),\min(C)\}\cap\{\min(A'),\min(B'),\min(C')\}$.

A lot of variants of the puzzle has been proposed since Lucas' orginal one. We report the reader to the highly valuable book \cite{Hinz2018} for a general synthesis on the subject.

Lucas already understood that the Tower of Hanoi was deeply linked to numeration systems. Indeed, he wrote in 1893 \cite[p. 58]{Lucas1893} that

\begin{quote}Increasing the number of pegs and slightly modifying the rule of the game would easily provide representations of all numeration systems. [{\em En augmentant le nombre de tiges et en modifiant légèrement les règles du jeu, on obtiendrait facilement des représentations de tous les systèmes de numération.}]
\end{quote}

The optimal algorithm to solve the puzzle with $n$ disks requires $2^n-1$ moves (hence passes through $2^n$ states), and the total number of admissible states is $3^n$. (There exists a ``worst'' algorithm that solves the puzzle  passing through all these $3^n$ state exactly once). Hence, it is not a surprise that there are natural links between the Tower of hanoi and binary and ternary numeration systems. At Lucas' time, only integral numeration systems were known. Since the ``worst'' solution (i.e. the solution that pass through all possible states) requires $3^n-1$ moves, it is sensible to ask for more pegs to represent other numeration systems. But now that non-integer numeration systems are known, we can give to Lucas' sentence a new meaning, keeping the three initial pegs and modifying the rules of the game, to get a Tower of Hanoi version of some nonconventional ways to write integers.

A first possibility consists in restricting the moves allowed between pegs. For example, we can forbid any direct move from $A$ to $C$ and from $C$ to $A$. It is well-known that this constraint leads to the ``worst'' algorithm mentioned beforehand, that visits every possible state of the puzzle among the $3^n$ ones. Each of the possible variants of this kind is linked to some numeration system (as well as to some Gray code) defined by a linear recurring sequence (see \cite{RittaudRome}). 

The initial question that gave rise to the present article was the converse: find natural alternative rules for the Tower of Hanoi such that the minimal number of moves required to solve the puzzle with $n$ disks corresponds to a sequence fixed {\em a priori}. One of these sequences for which an answer can be found is the Fibonacci sequence, and we will shw that the answer described here extends to some other linear recurring sequences as well. 

Apart from the sequence of minimal moves, other links between the Tower of Hanoi and the Fibonacci sequence can be made. In particular, it is shown in \cite{Hinz2019} that, for the classical puzzle with $n$ disks, the number of {\em key states} (i.e. for which the minimal number of moves to reach $(\varnothing,\varnothing,\Delta_n)$ is exactly twice the minimal number of moves to reach $(\Delta_n,\varnothing,\varnothing)$) is equal to $F_{n-1}$. The same article mentions the following other result, due to Merryfield and published in \cite{Bennish}: for any $n$, the set of distincts $A_k$ (resp. $B_k$, $C_k$) attained during the optimal algorithm of resolution of the standard puzzle is of cardinality $F_{n+2}$ (resp. $F_{n+1}$, $F_{n+2}$).

The present paper is organized as follows. In Section \ref{HanoiClassique}, we recall the relevant properties of the classical Tower of Hanoi we wish to generalize. Section \ref{HanoiFibo} is devoted to the main variant we are interested in. In this variant is defined {\em Fibonacci moves}. We prove that the variant defined by these kind of moves, the {\em Tower of Hanoi-Fibonacci}, is optimally solved in a number of moves essentially given by the Fibonacci sequence (Section \ref{Zeck}). Then, we provide a link with the classical Zeckendorf-Fibonacci numeration system (Section \ref{Zeck}) and deduce from it an iterative algorithm for the Tower of Hanoi-Fibonacci. We then study a Gray-like code associated to this numeration system (Section \ref{Gray}), then investigate the general properties of the graph associated to the puzzle (Section \ref{Graph}). Eventually, in Section \ref{Gene} are briefly investigated some generalizations and questions, in two directions. The first one is when the definition of Fibonacci moves is modified so as to get an optimal algorithm that requires a number of moves given by a linear recurring sequence of the form $m_n=m_{n-p}+m_{n-q}+1$. The second one considers complementary restrictions on moves between pegs, that gives rise to a Tribonacci-like sequence.

\section{The classical Tower of Hanoi}\label{HanoiClassique}

In the following, a subset $\{d_{k_1},\ldots, d_{k_i}\}$ of $\Delta_n$ with $k_1<\cdots<k_i$ is simply written $k_1\cdots k_i$.

 It is known since Lucas that the Tower of Hanoi has a solution for any  $n\geqslant 0$, and that there is a unique solution with minimal number $m_n$ of moves. Such a solution can be described recursively by the following decomposition, valid for any $n\geqslant 1$, from which we easily deduce that $m_n=2m_{n-1}+1$, hence $m_n=2^n-1$ (since $m_0=0$) .
 \[(\Delta_n,\varnothing,\varnothing)\stackrel{m_{n-1}}{\longrightarrow}(n,\Delta_{n-1},\varnothing)\stackrel{1}{\longrightarrow}(\varnothing,\Delta_{n-1}, n)\stackrel{m_{n-1}}{\longrightarrow}(\varnothing,\varnothing,\Delta_n).\]
 
Since the number of states during the optimal resolution of the puzzle is $1+m_n=2^n$, it is natural to consider the binary expansion of length exactly $n$ to code the successive states from $0$ to $2^n-1$. It is easily proved by induction that the index $k$ (between $1$ to $n$) of the leftmost changing digit from the binary expansion of $i$ to the binary expansion of $i+1$ corresponds to the disk which is moved when going from the state $i$ to the state $i+1$. As a consequence, we have that, for any $1\leqslant k\leqslant n$, the number of times the disk $d_k$ is moved is $2^{n-k}$.
 
Also, consider the alternative codage of the states of the puzzle with $n$ disks by elements of $\{0,1\}^n$ given by the following rules: the initial state $(\Delta_n,\varnothing,\varnothing)$ is labelled $0^n$, and when we go from partition $(A,B,C)$ to $(A',B',C')$ by moving the disk $d=d_k$, the label of the state $(A',B',C')$ is defined as the label of $(A,B,C)$ in which the $k$-th digit has been switched (that is: this digit becomes a $0$ if it was a $1$ and a $1$ if it was a $0$). Then, an induction shows that the sequence of codages of the successive states thus obtained coincides with the {\em reflected binary Gray code}, that is: the list $\mathcal{G}_n$ of all binary words with exactly $n$ letters defined recursively by $\mathcal{G}_0=\{0\}$ and $\mathcal{G}_n=0\mathcal{G}_{n-1}+1\overline{\mathcal{G}_{n-1}}$ (where, for a sequence $\mathcal{L}=\{x_1,\ldots, x_k\}$ of words and $d$ a letter, $d\mathcal{L}=\{dx_1,\ldots, dx_n\}$ and, with $\mathcal{L}'=\{y_1,\ldots, y_\ell\}$, the notation $\mathcal{L}+\mathcal{L}'$ stands for $\{x_1,\ldots, x_k,y_1,\ldots y_\ell\}$). The fundamental property of such a list $\mathcal{G}_n$ is that any two consecutive elements of the list differ by exactly one digit.

For a given state, when only the partition is known but not the order of its elements, we write the partition as $X\sqcup Y\sqcup Z$. An expression like $X\sqcup Y\sqcup Z\longrightarrow R\sqcup S\sqcup T$ means a move (or a sequence of moves) in which the position of the element $R$ (resp. $S$, $T$) of the final partition is the same as the position of the element $X$ (resp. $Y$, $Z$) in the initial partition. It was observed in \cite{Hinz1992} that the graph $\mathcal{H}_n=(V_n,E_n)$ of the classical Tower of Hanoi has a fractal structure similar to the Sierpi\'nski triangle, $\mathcal{H}_{n}$ being made of three copies of $\mathcal{H}_{n-1}$ for any $n\geqslant 1$, any two of these copies being linked by a single edge corresponding to a move of the form $n\sqcup\Delta_{n-1}\sqcup\varnothing\longrightarrow\varnothing\sqcup\Delta_{n-1}\sqcup n$ (see Figure \ref{GrapheClassique}).

\begin{figure}
\centering\includegraphics[height=9cm]{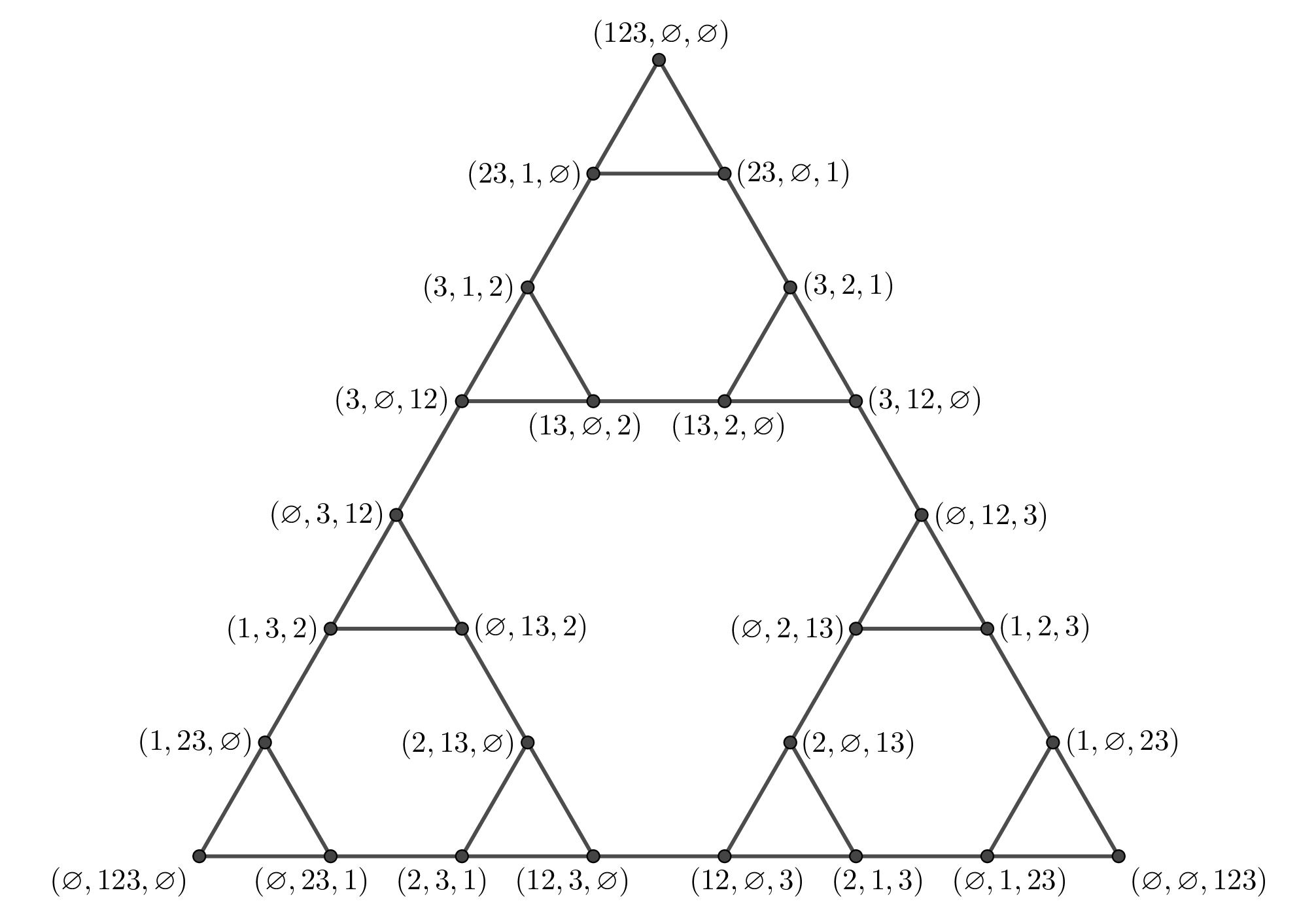}
\begin{caption}{The graph $\mathcal{H}_3$ of the classical Tower of Hanoi with $3$ disks.}\end{caption}
\label{GrapheClassique}
\end{figure}

Eventually, the optimal solution of the puzzle can be described by the following algorithm: move $d_1$ (always in the way $A\rightarrow B\rightarrow C\rightarrow A$ if $n$ is odd, and in the way $A\rightarrow C\rightarrow B \rightarrow A$ if $n$ is even), then, while there is a disk $d_k\neq d_1$ that can be moved, move that disk then then move again $d_1$.

\section{The Tower of Hanoi-Fibonacci}\label{HanoiFibo}

\subsection{Definition and optimal algorithm}\label{DefEtAlgo}

\begin{definition}\label{FibonacciRule} Let $X$ and $Y$ be two different pegs such that, for some $k\in\Delta_n$, we have $X=k\tilde{X}$ and $Y=\Delta_{k-1}\tilde{Y}$. Write $Z$ for the third peg. We define a {\em $k$-Fibonacci move} as a move that consists in putting simultaneously both $k-1$ and $k$ onto $Z$, i.e.:
\[k\tilde{X}\sqcup \Delta_{k-1}\tilde{Y}\sqcup Z\longrightarrow \tilde{X}\sqcup \Delta_{k-2}\tilde{Y}\sqcup (k-1)kZ.\]

A {\em Fibonacci move} is a $k$-Fibonacci move for some $k$.  The {\em Tower of Hanoi-Fibonacci} is the Tower of Hanoi puzzle in which only Fibonacci moves are allowed. (Note that this definition will be slightly modified in Section \ref{Graph}.)
\end{definition}

\begin{figure}
\centering\includegraphics[height=1.8cm]{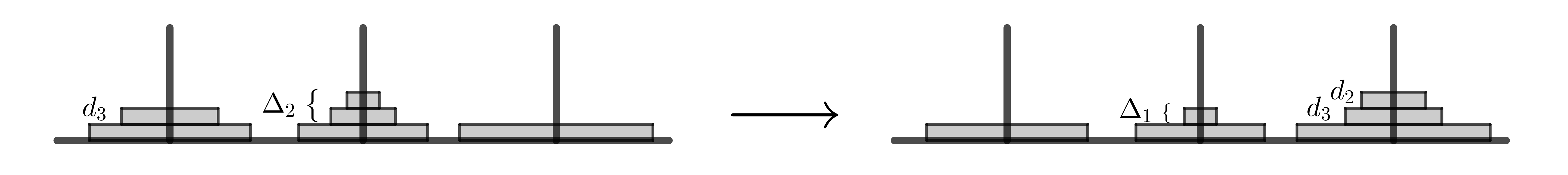}
\begin{caption}{The $3$-Fibonacci move from $(35,\Delta_24,6)$ to $(5,\Delta_14,236)$.}\end{caption}
\label{FibMove}
\end{figure}

Note that the $1$-Fibonacci move is the one in which only one disk is moved (the disk $d_1$). Hence, this move is the only one common to the Tower of Hanoi-Fibonacci and the classical Tower of Hanoi.


\begin{theorem}\label{OptimalFibo} The Tower of Hanoi-Fibonacci with $n$ disks admits a solution for any $n\geqslant 0$. There is only one optimal algorithm for it, that needs exactly $F_{n+2}-1$ Fibonacci moves (hence passes through $F_{n+2}$ different states).
\end{theorem}

As an example, here is the optimal solution in the case $n=5$. 
\[(\Delta_5,\varnothing,\varnothing)\longrightarrow(2345,\varnothing, 1)\longrightarrow(345, 12,\varnothing)\longrightarrow (45, 1, 23)\longrightarrow\]
\[(45,\varnothing, 123)\longrightarrow(5, 34, 12)\longrightarrow (15, 34, 2)\longrightarrow(5, 1234,\varnothing)\longrightarrow\]
\[(\varnothing,123, 45)\longrightarrow(\varnothing, 23, 145) \longrightarrow(12, 3, 45)\longrightarrow(1,\varnothing, 2345)\longrightarrow(\varnothing,\varnothing,\Delta_5)\]

\begin{proof} The proof is very similar to the classical case. First, the cases $n=0$ and $n=1$ both admit trivial solutions, with $m_0=0=F_2-1$ and $m_1=1=F_3-1$, where $m_n$ stands for the minimal number of moves to solve the puzzle with $n$ disks. Now, put $n\geqslant 2$ and assume that the puzzle with $n-1$ and $n-2$ disks are both solvable, with $m_{n-2}=F_{n}-1$ and $m_{n-1}=F_{n+1}-1$. Consider the puzzle with $n$ disks. To be moved, the disk of radius $n$ needs to be alone on its peg, and needs the tower $\Delta_{n-1}$ to be on another peg. To mimimize the number of moves, we can ask $d_n$ to be moved only once, hence any solution of the puzzle needs to reach the state $(n,\Delta_{n-1},\varnothing)$. We thus get a recursive description of the optimal solution of the puzzle with $n$ disks:
\[(\Delta_n,\varnothing,\varnothing)\stackrel{m_{n-1}}{\longrightarrow}(n,\Delta_{n-1},\varnothing)\stackrel{1}{\longrightarrow}(\varnothing,\Delta_{n-2}, (n-1)n)\stackrel{m_{n-2}}{\longrightarrow}(\varnothing,\varnothing,\Delta_n).\]

Hence, we have that $m_n=m_{n-1}+1+m_{n-2}$, so, by the induction hypothesis, $m_n=(F_{n+1}-1)+1+(F_{n}-1)=F_{n+2}-1$, as required.\end{proof}

\subsection{Link with the Zeckendorf-Fibonacci numeration system}\label{Zeck}

The classical link between binary numeration system and the standard Tower of Hanoi  extends to the {\em Zeckendorf} (or {\em Zeckendorf-Fibonacci}) {\em numeration system} and the Tower of Hanoi-Fibonacci. This link will provide us an iterative algorithm for the optimal solution of the latter puzzle.

Recall that, as proved in \cite{Zeckendorf}, for any fixed $n\geqslant 2$ and any integer $0<k<F_{n+1}$ there exists a unique finite sequence $(u_i)_{2\leqslant i\leqslant n}\in\{0,1\}^{n-1}$ such that $u_iu_{i+1}=0$ for any $i$ and ${\displaystyle k=\sum_{2\leqslant i\leqslant n} u_iF_{i}}$. Such a sequence is the {\em Zeckendorf-Fibonacci expansion} of $k$. We will also write $[u_{n}\cdots u_2]_F$ for it, and $Z((u_i)_{2\leqslant i\leqslant n})$ for the corresponding $k$. By convention, we may define $Z(0)$ as the empty sequence. When we need the length of a Zeckendorf-Fibonacci expansion to be of a certain kind (as in the following theorem), we allow ourselves to append leading $0$s to it, hence considering $[0u_n\cdots u_2]_F$ as equivalent to $[u_n\cdots u_2]_F$.

In the following, a {\em ZF-sequence} (or {\em ZF-word}) will denote any binary sequence (or binary word) satisfying the property that it does not contains two successive $1$ anywhere in its terms. Two ZF-words like $u_n\cdots u_2$ and $0u_n\cdots u_2$ will be regarded as equivalents. Under this equivalence relation, the Zeckendorf-Fibonacci expansion of $k$ is unique. Moreover, this expansion defines a bijection from $\N^*$ onto the set of (non-empty) ZF-sequences.

The Zeckendorf-Fibonacci expansion of  $k>0$ can be obtained by the application of the following algorithm:

\begin{itemize}

\item $r:=k$, $u_i:=0$ for all $i\geqslant 2$, $n:=\max(j\geqslant 2\ :\ F_j\leqslant k)$ 

\item while $r>0$:

\indent\indent $i:=\max(j\ :\ F_j\leqslant r)$

\indent\indent $u_i:=1$

\indent\indent $r:=r-F_i$

\item return$((u_i)_{2\leqslant i\leqslant n})$.

\end{itemize}

\begin{theorem}\label{NumerationFibo} Let $n\geqslant 0$ be some integer, let $0<k<F_{n+2}$, and let $k-1=[u_{n+1}\cdots u_2]_F$ and $k=[v_{n+1}\cdots v_2]_F$. Let $j$ be the biggest index such that $u_j\neq v_j$. The $k$-th move of the Tower of Hanoi-Fibonacci with $n$ disks is a $(j-1)$-Fibonacci move.
\end{theorem}

\begin{proof} This is a simple induction making use of the recursive description of the algorithm
\[(\Delta_n,\varnothing,\varnothing)\longrightarrow(n,\Delta_{n-1},\varnothing)\longrightarrow(\varnothing,\Delta_{n-2}, (n-1)n)\longrightarrow(\varnothing,\varnothing,\Delta_n).\]

The property is true for $n=0$ and $n=1$. Assume it is true for $n-2$ and $n-1$ for some $n\geqslant 2$.  The moves from $(\Delta_n,\varnothing,\varnothing)$ to $(n,\Delta_{n-1},\varnothing)$ are moves from $1$ to $m_{n-1}=F_{n+1}-1$, so their Zeckendorf-Fibonacci exansion of length $n$ are all of the form $[0u_{n}\cdots u_2]_F$. Hence, by induction hypothesis on the puzzle with $n-1$ disks, the property is true for all these moves.

Now, the Fibonacci move $(n,\Delta_{n-1},\varnothing)\longrightarrow(\varnothing,\Delta_{n-2}, (n-1)n)$ is the $F_{n+1}$-th one, of Zeckendorf-Fibonacci expansion of length $n$ equal to $[10\cdots 0]_F$. The biggest moving disk in this move is the $n$-th one, and the biggest index $j$ as defined in the theorem is equal to $n+1$, so the theorem is also valid for this move. 

The remaining $F_{n}$ moves are the ones with Zeckendorf-Fibonacci expansion of length $n$ of the form $[10u_{n-1}\cdots u_2]_F$, where $[u_{n-1}\cdots u_0]_F$ is the Zeckendorf-Fibonacci expansion of length $n-2$ of $k-F_{n+1}$. Hence, we can apply to it the induction hypothesis made on the puzzle with $n-2$ disks, and we are done.\end{proof}

\begin{theorem}\label{NbMovesDiskK} Let $0<k\leqslant n$ be two integers. For the Hanoi-Fibonacci puzzle with $n$ disks, the number of $k$-Fibonacci moves in the optimal algorithm is equal to $F_{n+1-k}$.
\end{theorem}

\begin{proof} We procede by induction on $n$ and (decreasing) induction on $k$. For $k=n$, the number we are looking for is equal to $1$, which corresponds indeed to $F_{n+1-k}=F_{1}=1$. For $k=n-1$, it is easy to check that there is also exactly one $k$-Fibonacci move, a number equal to $F_{n+1-k}=F_2$. Now, for $k<n-1$, by induction hypothesis, the part $(\Delta_n,\varnothing,\varnothing)\longrightarrow(n,\Delta_{n-1},\varnothing)$ of the algorithm involves a number of $k$-Fibonacci moves equal to $F_{(n-1)+1-k}=F_{n-k}$, and the part $(\varnothing,\Delta_{n-2}, (n-1)n)\longrightarrow(\varnothing,\varnothing,\Delta_n)$ involves $F_{(n-2)+1-k}=F_{n-k-1}$ moves, hence a total equal to $F_{n-k}+F_{n-k-1}=F_{n+1-k}$ moves.\end{proof}

Theorem \ref{NumerationFibo} provides a complete iterative algorithm for the Tower of Hanoi-Fibonacci, with the only issue that, for $1$-Fibonacci moves (i.e. a move of the single disk $d_1$), one has to determine on which peg the disk $d_1$ has to move. Here is an answer to this question. Let us say that $d_1$ is moving to the right (resp. to the left) whenever it moves from $A$ to $B$, from $B$ to $C$ or from $C$ to $A$ (resp. from $A$ to $C$, from $B$ to $A$ or from $C$ to $B$). Thus, by Theorem \ref{NbMovesDiskK}, we can code the sequence of $1$-Fibonacci moves for the puzzle with $n$ disks as a word $\mu_n\in\{l,r\}^{F_n}$, where $l$ denotes a move to the left and $r$ a move to the right. We then have the following result:

\begin{theorem}\label{SubstBizarre} In the optimal algorithm for the Tower of Hanoi-Fiboonacci:

\begin{itemize}

\item if $n\in 2\N^*$, then the $k$-th letter of $\mu_n$ is a $r$ iff $Z(k)$ has an even number of $1$s;

\item if $n\notin 2\N^*$, then the $k$-th letter of $\mu_n$ is a $r$ iff $Z(k)$ has an odd number of $1$s.

\end{itemize}
\end{theorem}

\begin{proof} We proceeds by induction on $n\geqslant 2$. Write the decomposition of the optimal solution of the puzzle, with the corresponding number of $1$-Fibonacci moves (given by Theorem \ref{NbMovesDiskK} with $k=1$).
\[(\Delta_n,\varnothing,\varnothing)\stackrel{F_{n-1}-1}{\longrightarrow}(n,\Delta_{n-1},\varnothing)\stackrel{0}{\longrightarrow}(\varnothing,\Delta_{n-2}, (n-1)n)\stackrel{F_{n-2}-1}{\longrightarrow}(\varnothing,\varnothing,\Delta_n).\]
Assume for example $n\in 2\N^*$ (the other case would be similar). Consider the $k$-th letter of $\mu_n$ corresponding to a $1$-Fibonacci move among the $F_{n-1}-1$ first ones. By the induction hypothesis, the fact that $n-1$ is odd and the fact that the $1$-Fibonacci moves corresponding to $(\Delta_n,\varnothing,\varnothing)\longrightarrow(n,\Delta_{n-1},\varnothing)$ are the same as the one for $(\Delta_n,\varnothing,\varnothing)\longrightarrow (n,\varnothing,\Delta_{n-1})$ but with exchanging the $r$s and the $l$s, we have that the considered $k$-th letter is a $r$ iff $Z(k)$ has an even number of $1$s. For a value of $k$ corresponding to a $1$-Fibonacci move among the last $F_{n-2}-1$ Fibonacci moves, the reasoning is the same, with the additional consideration that the Zeckendorf expansion of $k$ is now of the form $[10u_{n-3}\cdots u_2]_F$.
\end{proof}

Let us also mention the following qualitative results, that show in particular that the number of $1$-Fibonacci moves to the left and to the right are as balanced as possible.

\begin{corollary}\label{Qualit} We have $\mu_0=\varnothing$, $\mu_1=l$ and, for any $n\geqslant 2$, $\mu_n=(\mu_{n-1}\mu_{n-2})^*=\mu_{n-2}\mu_{n-3}\mu_{n-3}\mu_{n-4}$, where $\mu^*$ stands for the word in which each $l$ has been replaced by a $r$ and each $r$ by a $l$. Moreover, denoting by $|\mu|_d$ the number of letter $d\in\{r,l\}$ in the word $\mu$, we  have
\[|\mu_{3n}|_r=|\mu_{3n}|_l,\qquad |\mu_{3n+1}|_l-|\mu_{3n+1}|_r=1\qquad\text{and}\qquad |\mu_{3n+2}|_r-|\mu_{3n+2}|_l=1.\]
\end{corollary}

\begin{proof}

The first part is proved by an induction on $n$ and the following decomposition of the optimal solution of the puzzle (for $n\geqslant 2$), in which the words on the arrows stand for the sequences of moves of $d_1$ during $1$-Fibonacci moves:
\[(\Delta_n,\varnothing,\varnothing)\stackrel{\mu_{n-1}^*}{\longrightarrow}(n,\Delta_{n-1},\varnothing)\stackrel{}{\longrightarrow}(\varnothing,\Delta_{n-2}, (n-1)n)\stackrel{\mu_{n-2}^*}{\longrightarrow}(\varnothing,\varnothing,\Delta_n).\]
The second part is a simple induction on $n$.
\end{proof}

\subsection{Gray-like code}\label{Gray}

To complete notations set up for lists in Section \ref{HanoiClassique}, for any list $\mathcal{L}$ of elements of $\{0,1\}^n$, write $\mathcal{L}'$ for the list made of all elements of $\mathcal{L}$ in which the leftmost letter is removed. Then, set $\mathcal{N}_0=\varnothing$, $\mathcal{N}_1:=\{1\}$ and $\mathcal{N}_n:=10\overline{\mathcal{N}_{n-1}'}+ 10\overline{\mathcal{N}_{n-2}}$ (here writing each word always with a $1$ as leftmost digit). Eventually, let ${\displaystyle\mathcal{G}:=\sum_{i\geqslant 0}\mathcal{N}_i:=\{g_1,g_2,\ldots\}}$. Such a construction may be seen as a mirroring process analogous to the classical one for binary Gray codes, as shown in Table \ref{GrayFiboHanoi6}.

\begin{table}[h!]
\centering
\begin{tabular}{rrcccccccl}
\ldelim\{{20}{7mm}[$\mathcal{G}_6$ ]&$g_1=$&0&0&0&0&0&1&\rdelim\}{1}{5mm}[ $\mathcal{N}_1$]\\
\cdashline{7-8}[2pt/2pt]
&$g_2=$&0&0&0&0&1&0&\rdelim\}{1}{5mm}[ $\mathcal{N}_2$]\\
\cdashline{6-7}[2pt/2pt]\cline{8-8}
&$g_3=$&0&0&0&1&0&0&\rdelim\}{2}{5mm}[ $\mathcal{N}_3$]\\
&$g_4=$&0&0&0&1&0&1\\
\cdashline{5-6}[2pt/2pt]\cline{7-8}
&$g_5=$&0&0&1&0&0&1&\rdelim\}{3}{5mm}[ $\mathcal{N}_4$]\\
&$g_6=$&0&0&1&0&0&0\\
&$g_7=$&0&0&1&0&1&0\\
\cdashline{4-5}[2pt/2pt]\cline{6-8}
&$g_8=$&0&1&0&0&1&0&\rdelim\}{5}{5mm}[ $\mathcal{N}_5$]\\
&$g_9=$&0&1&0&0&0&0\\
&$g_{10}=$&0&1&0&0&0&1\\
&$g_{11}=$&0&1&0&1&0&1\\
&$g_{12}=$&0&1&0&1&0&0\\
\cdashline{3-4}[2pt/2pt]\cline{5-8}
&$g_{13}=$&1&0&0&1&0&0&\rdelim\}{8}{5mm}[ $\mathcal{N}_6$]\\
&$g_{14}=$&1&0&0&1&0&1\\
&$g_{15}=$&1&0&0&0&0&1\\
&$g_{16}=$&1&0&0&0&0&0\\
&$g_{17}=$&1&0&0&0&1&0\\
&$g_{18}=$&1&0&1&0&1&0\\
&$g_{19}=$&1&0&1&0&0&0\\
&$g_{20}=$&1&0&1&0&0&1\\
&&&&&&&
\end{tabular}
\caption{The Gray-like code of the Tower of Fibonacci-Hanoi with $n=6$ disks.}
\label{GrayFiboHanoi6}
\end{table}



Recall that the {\em Hamming distance} between $w$ and $w'\in\{0,1\}^n$, written $h(w,w')$, is the number of their differents digits, that is: for $w=w_1\cdots w_n$ and $w'=w'_0\cdots w'_n$, we have $h(w,w')=\card(1\leqslant i\leqslant n\ :\ w_i\neq w'_i)$. When $w$ and $w'$ do not have the same number of letters, we append as many $0$ as necessary to the shortest one to make it of the same length as the other.

\begin{theorem}\label{QuasiHamming}  The application $g_n\longmapsto Z(n)$ is a bijection from the set of all non-empty ZF-words (assuming the equivalence between $u_n\cdots u_2$ and $0u_n\cdots u_2$) onto $\N^*$. Moreover, for any $n\geqslant 1$ we have
\[h(g_n,g_{n+1})=\left\{\begin{array}{cl}2&\mbox{if $n+1=F_{k}$ for some $k\geqslant 3$;}\\1&\mbox{otherwise.}\end{array}\right.\]

\end{theorem}

\begin{proof} By induction, assume that $\mathcal{N}_{n-1}'\cup\mathcal{N}_{n-2}$ contains all ZF-words of $\{0,1\}^{n-2}$ exactly once, $\mathcal{N}_{n-1}'$ (resp. $\mathcal{N}_{n-2}$) containing those with a $0$ (resp. a $1$) as leftmost digit. Hence, by definition of $\mathcal{N}_n$, $\mathcal{N}_n$ contains all ZF-words in $\{0,1\}^n$ with $10$ as leftmost digits, each exactly once. As a consequence, by the induction hypothesis, $\mathcal{N}_n'\cup\mathcal{N}_{n-1}$ contains all ZF-words of $\{0,1\}^{n-1}$ exactly once, with $\mathcal{N}_n'$ (resp. $\mathcal{N}_{n-1}$) containing those with a $0$ (resp. a $1$) as  leftmost digit. Hence, $\mathcal{G}$ is in bijection with the set of all ZF-words (but the empty one). The theory of Zeckendorf-Fibonacci recalled in Section \ref{Zeck} thus gives us that $g_n\longmapsto Z(n)$ is a bijection onto $\N^*$.

Now for the property of Hamming distances. By induction, assume the property satisfied for $\mathcal{G}_{n-1}$ and that the Hamming distance among two successive elements of $\mathcal{N}_k$ is always equal to $1$. We have $\mathcal{G}_n=0\mathcal{G}_{n-1}+\mathcal{N}_n=0\mathcal{G}_{n-1}+10\overline{\mathcal{N}_{n-1}'}+10\overline{\mathcal{N}_{n-2}}$. Hence, by induction hypothesis, the Hamming distance of any pair of successive elements of $\mathcal{G}_n$ is equal to $1$, apart, possibly, for the pairs $(0g_{F_{n+1}-1},g_{F_{n+1}})$ (made by the last element of $\mathcal{G}_{n-1}$ and the following one in $\mathcal{G}_n$) and $(g_{F_{n+1}+F_{n-1}-1},g_{F_{n+1}+F_{n-1}})$ (made by the last element of $\mathcal{G}_{n-1}+10\overline{\mathcal{N}_{n-1}'}$ and the following one in $\mathcal{G}_n$).

To prove that the induction hypothesis is valid also for $n$, it only remains to prove that the Hamming distance of the first of these pairs is equal to $2$ and the one of the second pair is equal to $1$. For the first one, by definition of $\mathcal{N}_{n}$, the two different digits of $0g_{F_{n+1}-1}$ and $g_{F_{n+1}}$ are the two leftmost ones, that is: writing $0g_{F_{n+1}-1}$ in the form $010u_{n-3}\cdots u_1$, we have $g_{F_{n+1}}=100u_{n-3}\cdots u_1$. Hence their Hamming distance is indeed equal to $2$. For the second pair, by construction we have (writing $w'$ for the word $w$ from which the leftmost digit has been removed) $g_{F_{n+1}+F_{n-1}-1}=10g'_{F_{n+1}-F_{n-1}}=10g'_{F_n}$ and $g_{F_{n+1}+F_{n-1}}=10g_{F_n-1}$. As already noticed, $g_{F_n}=0100u_{n-4}\cdots u_1$ and $g_{F_n-1}0010u_{n-4}\cdots u_1$, so $10g'_{F_n}=1000u_{n-4}\cdots u_1$ and $10g_{F_n-1}=1010u_{n-4}\cdots u_1$, hence their Hamming distance is equal to $1$.\end{proof}

Strictly speaking, a Gray code has the property that two consecutive elements are always of Hamming distance equal to $1$, hence our list $\mathcal{G}$ is only a Gray-like code. One may wonder if we could recover a real Gray code that lists all the ZF-words. In itself, such a question is too large, and a natural restriction on it is to ask for such a Gray code to be {\em length-increasing}, that is: the list $(g_n)_{n\geqslant 1}$ should order the words in such a way that the leftmost $1$ of $g_n$ is of increasing index with $n$. It is quite easy to check that such a natural condition cannot be satisfied, since for $w$ the last ZF-word of length $n-1$ and $w'$ the first one of length $n$ we necessarily have $h(w,w')\geqslant 2$. This remark leads to the following result:

\begin{theorem}\label{AsFewAsPossible} Let $(u_n)_{n>0}$ be a length-increasing sequence made of all non-null ZF-words (each of them appearing exactly once). For any $n>0$, we have $h(u_n,u_{n+1})\geqslant h(g_n,g_{n+1})$. 
\end{theorem}

In this sense, our Gray-like code $\mathcal{G}$ is as close as possible to a true Gray code length-increasing and containing all the ZF-words.




As it is done in \cite[Theorem 5]{Rittaud2022} for another Gray-like code linked to Fibonacci combinatorics given in \cite{Bernini2013}, it is possible to ``de-mirror'' the construction of $\mathcal{G}_n$, in the following way:

\begin{theorem}\label{Demirror} Let $\mathcal{N}_0=\varnothing$, $\mathcal{N}_1=\{1\}$, $\mathcal{N}_2=\{10\}$ and $\mathcal{N}_3=\{100,101\}$. For any $n\geqslant 4$, all the following is well-defined and correct by induction. First, $\mathcal{N}_{n-1}$ has cardinality $F_{n-1}$. Set $\mathcal{N}_{n-1}=\{w_0,\ldots, w_{F_{n-1}-1}\}$. For any $0\leqslant m<F_{n-1}$, there exists a unique $q=q_n(m)\in\{-2,-1,0,1,2\}$ of smallest absolute value such that $w_{m+q}$ ends with a $1$. Eventually, for any $0\leqslant m<F_{n-1}$, set
\[\mathcal{V}_m=\begin{cases}\{w_i0\}&\text{if $q_n(m)=0$,}\\
\{w_i0,w_i1\}&\text{if $q_n(m)\in\{-1,2\}$,}\\
\{w_i1,w_i0\}&\text{if $q_n(m)\in\{-2,1\}$.}\end{cases}\]

We then have ${\displaystyle \mathcal{N}_n=\sum_{0\leqslant m<F_n}\mathcal{V}_m}$.
\end{theorem}

Practically speaking, the previous result can be enforced in the following way: to get $\mathcal{N}_n$ from $\mathcal{N}_{n-1}$, write each element of the list $\mathcal{N}_{n-1}$, each element being written twice in a row iff it ends with a $0$. Then, concatenate (on the right) to each element of this new list either a $0$ or a $1$ so as to get a word that differs from the previous one by exactly one digit (or, for the first of the list, exactly two digits from the last word of $\mathcal{N}_{n-1}$). 

For the proof, even if the Gray code in \cite{Bernini2013} does not exactly match the one considered in the present article, the two situations are close enough to allow us to invoke here the proof given in \cite{Rittaud2022} for Theorem \ref{Demirror} since, as can be easily checked, the sequence of last digits (starting at $\mathcal{U}_2$) are exactly opposite (i.e. with $1$s and $0$s exchanged) in both definitions. Following \cite{Rittaud2022} also leads to the following description that makes use of the Fibonacci substitution to describe the way digits are to be added to the elements of $\mathcal{N}_{n-1}$ to get $\mathcal{N}_n$. (Hence we do not provide the details here either, see \cite[Corollary 1]{Rittaud2022}.)

\begin{corollary}\label{FiboSubst} Let $\sigma$ be the Fibonacci substitution on the alphabet $\{\alpha,\beta\}$, defined by $\sigma(\alpha)=\alpha\beta$ and $\sigma(\beta)=\alpha$, and let $(\sigma_n)_{n\in\N}:=\sigma^{\infty}(\alpha)$ be its fixed point. Let $(\tau_n)_{n\in\N^*}$ be the word on the alphabet $\{0,1\}$ defined by $\tau_1=1$ and, for any $n\geqslant 2$, $\tau_n=1$ iff $\sigma_{\lfloor n/2\rfloor-1}=\beta$. To get $\mathcal{N}_{n+1}$ from $\mathcal{N}_n=\{g_{F_{n+1}},\ldots, g_{F_{n+2}-1}\}$ in Theorem \ref{Demirror} (with $n\geqslant 2$), we can proceed by the following algorithm:

\begin{itemize}

\item Initialization: $\mathcal{L}:=\varnothing$

\item for $i$ from $F_{n+1}$ to $F_{n+2}-1$:

\indent\indent if $g_i$ ends with a $1$ then $\mathcal{L}:=\mathcal{L}+\{g_i\}$

\indent\indent else $\mathcal{L}:=\mathcal{L}+\{g_i\}+\{g_i\}$

\item write $\mathcal{L}=:\{g_{F_{n+2}},\ldots, g_{F_{n+3}-1}\}$

\item for $i$ from $F_{n+2}$ to $F_{n+3}-1$:

\indent\indent in $\mathcal{L}$, replace $g_i$ by $g_i\tau_{i}$

\item return$(\mathcal{L})$.

\end{itemize}

\end{corollary}

\subsection{The Hanoi-Fibonacci graph}\label{Graph}

In the present section, our aim is to investigate how to represent the Tower of Hanoi-Fibonacci by a graph $\mathcal{F}_n=(V_n,E_n)$, in which $V_n$ as the $3^n$ possible states of the puzzle, and $E_n$ as the set of edges $(e,e')$ such that the move from $e$ to $e'$ is a Fibonacci move. Note that, contarily to the graph of the classical Tower of Hanoi, this graph is oriented since Fibonacci moves are not reversible (apart for $1$-Fibonacci moves).

\begin{theorem}\label{FiboNonPlanaire} For any $n\geqslant 2$, the un-oriented graph that corresponds to $\mathcal{F}_n$ is non-planar.
\end{theorem}

\begin{proof} Since $\mathcal{F}_n\subset \mathcal{F}_{n+1}$ for any $n$, it is sufficient to prove that $\mathcal{F}_2$ is non-planar.

\begin{figure}
\centering\includegraphics[height=8cm]{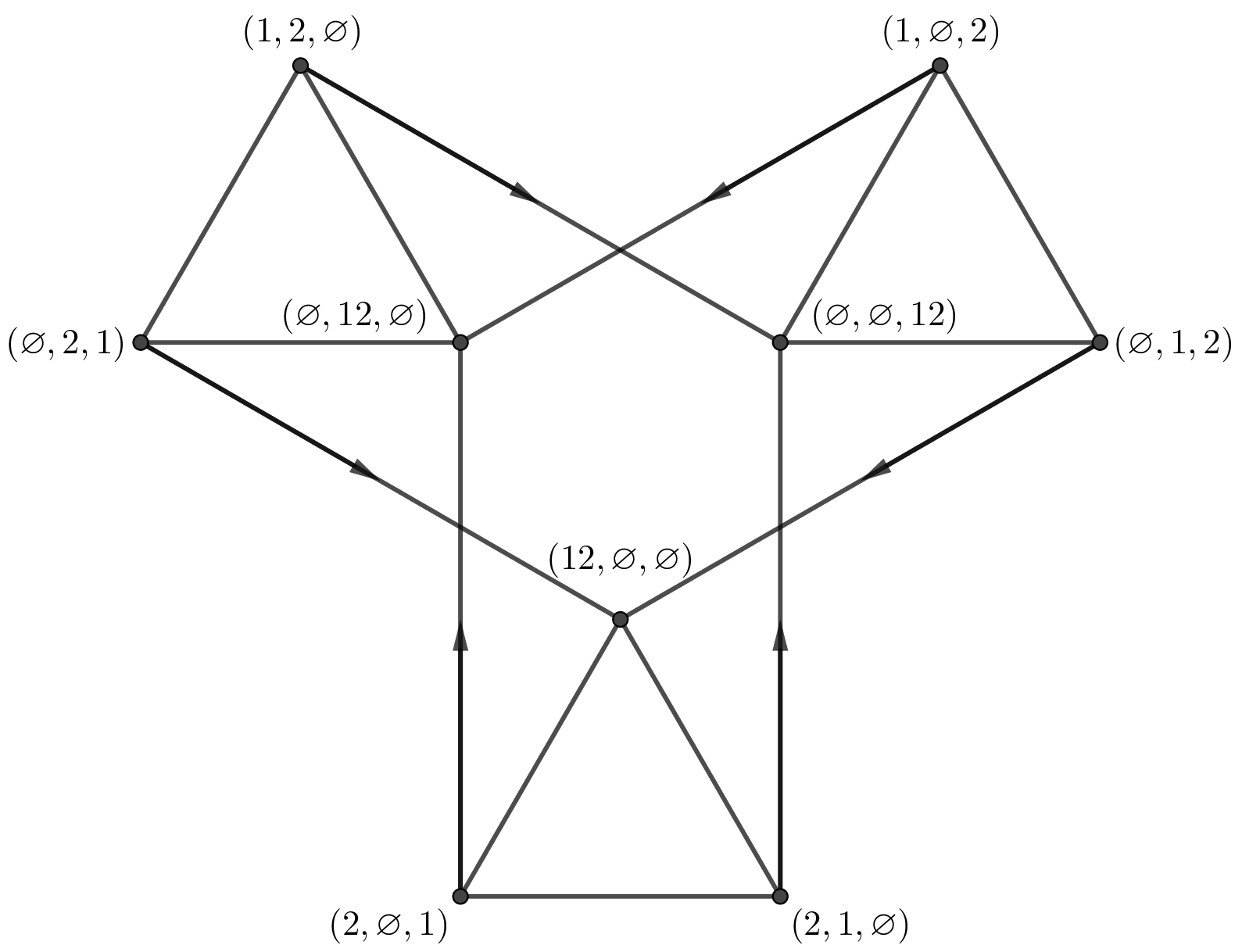}
\begin{caption}{The graph $\mathcal{F}_2$ of the Tower of Hanoi-Fibonacci with $2$ disks. (Line segments stand for edges going both ways.)}\end{caption}
\label{NaiveGraph}
\end{figure}

Let us merge the vertices $(\varnothing,1,2)$ and $(\varnothing,2,1)$ in a single vertex $a$, then the vertices $(1,\varnothing,2)$ and $(2,\varnothing,1)$ to get another single vertex $b$, and eventually the vertices $(1,2,\varnothing)$ and $(2,1,\varnothing)$ to get a third vertex $c$. The graph thus obtained can be split into two subsets of vertices, $V=\{a,b,c\}$ and $V'=\{(12,\varnothing,\varnothing),(\varnothing,12,\varnothing),(\varnothing,\varnothing,12)\}$. The set of edges of this new graph is the set of all possible edges between $V$ and $V'$, hence it is isomorphic to the complete bipartite graph $K_{3,3}$. Hence, $K_{3,3}$ is a minor of $\mathcal{F}_2$, so, by Wagner's theorem, $\mathcal{F}_2$ is not planar.\end{proof}

\begin{theorem}\label{FortementConnexe} For any $n\geqslant 0$, $\mathcal{F}_n$ is strongly connected. In other words, any possible state of the puzzle can be attained from any other under the  rule of the Tower of Hanoi-Fibonacci.
\end{theorem}

\begin{proof} We proceed by induction on $n$. Assume $\mathcal{F}_n$ is strongly connected, and consider $\mathcal{F}_{n+1}$. This latter graph contains exactly three copies of $\mathcal{F}_n$, that we denote by $\mathcal{F}_n^A$, $\mathcal{F}_n^B$ and $\mathcal{F}_n^C$ depending on the peg on which is located $d_{n+1}$ in each. By the induction hypothesis, each of these three copies of $\mathcal{F}_n$ is strongly connected. Hence, to obtain the desired result it is sufficient to prove that there exists an edge from some vertex of $\mathcal{F}_n^X$ to some vertex of $\mathcal{F}_n^Y$ for any different pegs $X$ and $Y$. The Fibonacci move $(n+1)\sqcup\Delta_n\sqcup\varnothing\longrightarrow\varnothing\sqcup\Delta_{n-1}\sqcup n(n+1)$ provides such an edge.\end{proof}

The previous drawing of $\mathcal{F}_2$ seems quite difficult to extend in a natural way to larger values of $n$, and a slight modification of Figure \ref{GrapheClassique} seems more interesting for visualization purposes, even if it needs some specific codage to makes the arrow diagram handy. Also, as announced in Definition \ref{FibonacciRule}, it will be here easier to work with a slightly modified version of a Fibonacci move, hereafter defined as 
\[k\tilde{X}\sqcup \Delta_{k-1}\tilde{Y}\sqcup Z\longrightarrow \Delta_{k-2}\tilde{X}\sqcup \tilde{Y}\sqcup (k-1)kZ.\]
In this new version, the tower $\Delta_{k-2}$ ends up on $\tilde{X}$ instead of remaining on $\tilde{Y}$. This does not fundamentally change what precedes, and it is easy to check that all the  results obtained under the initial Definition \ref{FibonacciRule} of Fibonacci moves remain unchanged under the present variant. (In particular, since the graph $\mathcal{F}_2$ remains the same for this variant, the new graph is still non-planar.)

Now, with this variant of Fibonacci moves, we can make use of the classical graph of Figure \ref{GrapheClassique} to represent the Hanoi-Fibonacci puzzle. We preserve in this graph the edges that represent $1$-Fibonacci moves. The other edges are also preserved, but in a form that we will call here {\em pseudo-edges}. More precisely: the graph $\mathcal{G}_n$ of the classical puzzle is made of three copies of $\mathcal{G}_{n-1}$, together with three edges that make $\mathcal{G}_n$ connected. In $\mathcal{F}_n$ we define these oriented edges as {\em $n$-pseudo-edges}, represented as arrows labelled by $2^{n-2}+1$ (for $n\geqslant 2$). Let $v$ be a vertex of $\mathcal{F}_n$ which is the origin of a $k$-pseudo-edge $e$ (hence with $1< k\leqslant n$). The $k$-Fibonacci move for the vertex $v$ ends up on the vertex $v'$ obtained by a jump of length $2^{k-2}+1$ in the direction of the edge, i.e. $v'$ is the vertex of the graph at a distance $2^{k-2}+1$ from $v$ (each edge or pseudo-edge counting for $1$) attained by following the path of length $2^{k-2}+1$ defined by the geometrical direction defined by $e$.

\begin{figure}
\centering\includegraphics[height=9cm]{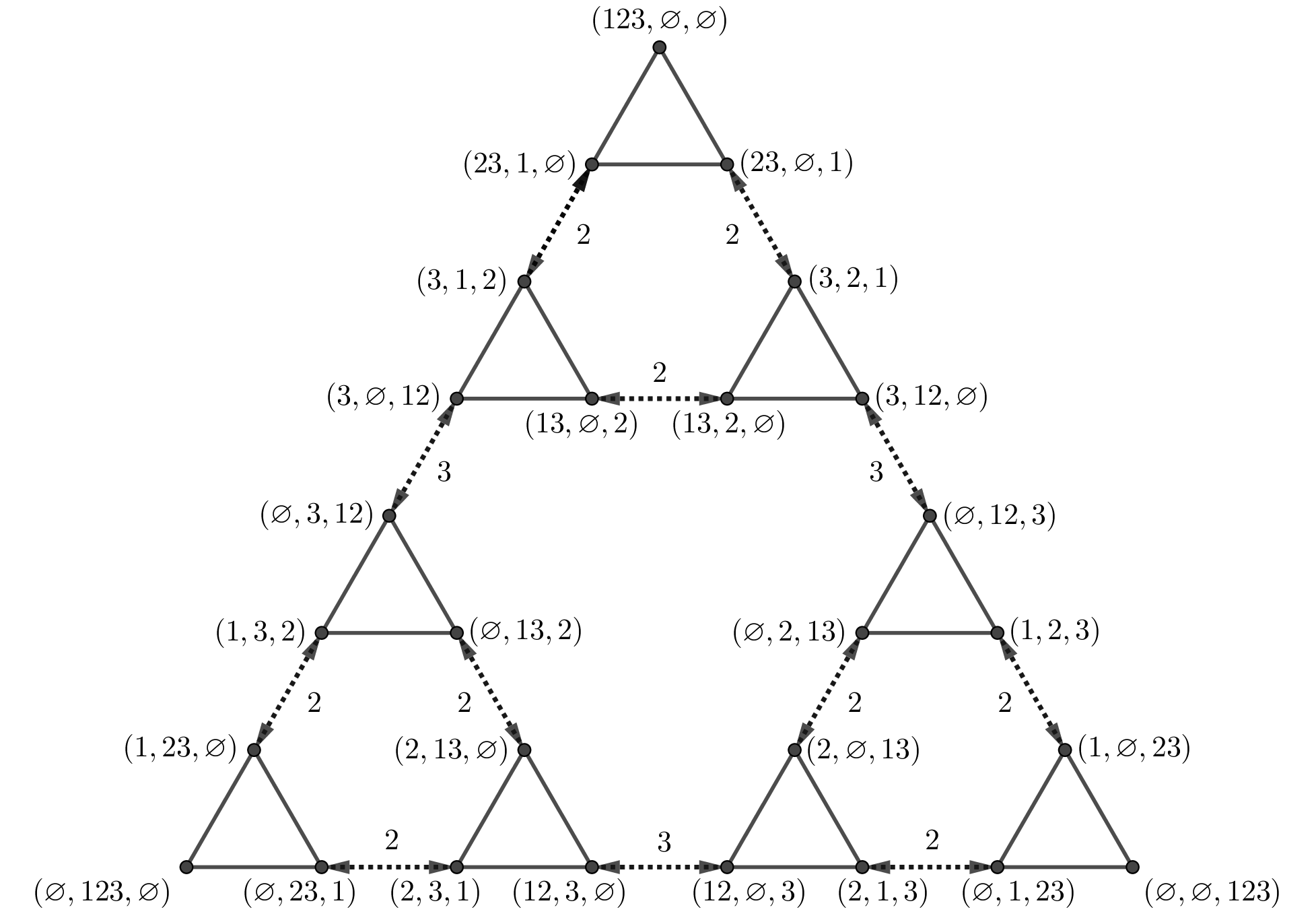}
\begin{caption}{The graph $\mathcal{F}_3$ of the Tower of Hanoi-Fibonacci with $3$ disks (under the variant of the Fibonacci moves) with its pseudo-edges.}\end{caption}
\label{GraphePasClassique}
\end{figure}

\begin{theorem}\label{CestBon} Under the previous definition of $\mathcal{F}_n$, if the vertex $v$ is the origin of a $k$-pseudo-edge ($k\geqslant 2$), then the vertex $v'$ is the state of the puzzle attained by the (only) possible $k$-Fibonacci move from the state $v$.
\end{theorem}

\begin{proof} Assume the result until $n-1$. The graph $\mathcal{F}_n$ contains three copies of $\mathcal{F}_{n-1}$, in each of which the property is true by induction. Therefore, it remains only to prove that the property is true also for the $n$-pseudo-edges of $\mathcal{F}_n$. By symmetry, it is enough to consider the case of the $n$-pseudo-edge of origin $(n,\Delta_{n-1},\varnothing)$. The $n$-Fibonacci move from this state leads to $(\Delta_{n-2},\varnothing,(n-1)n)$. Also, under the classical rules of the Tower of Hanoi puzzle, going from $(n,\Delta_{n-1},\varnothing)$ to $(\Delta_{n-2},\varnothing,(n-1)n)$ with the optimal algorithm requires exactly $2^{n-2}+1$ moves, which are all on the same geometrical direction on $\mathcal{H}_n$, so we are done.
\end{proof}

We deduce from this a combinatorial proof of the following equality.

\begin{corollary}\label{SommeMignonne} For any $n\geqslant 0$, we have
\[2^{n}=F_{n+2}+\sum_{k=0}^{n-2}2^{k}F_{n-1-k}.\]
\end{corollary}

\begin{proof} With the notation of Section \ref{HanoiClassique}, we have $2^n-1=m_n$. Also, by Theorem \ref{NbMovesDiskK} and the proof of Theorem \ref{CestBon}, we have 
\begin{eqnarray*}
m_n&=&F_{n}+\sum_{k=3}^{n+1}(2^{k-3}+1)F_{n+2-k}\\
&=&\sum_{k=2}^{n+1}F_{n+2-k}+\sum_{k=3}^{n+1}2^{k-3}F_{n+2-k}\\
&=&F_{n+2}-1+\sum_{k=0}^{n-2}2^{k}F_{n-1-k}.
\end{eqnarray*}
\end{proof}

Corollary \ref{SommeMignonne} can be related to the classical fact that the sum of the $n$-th row of Pascal's triangle is $2^n$ and the sum of its $n$-th diagonal is $F_{n}$.

\begin{figure}
\centering\includegraphics[height=6.5cm]{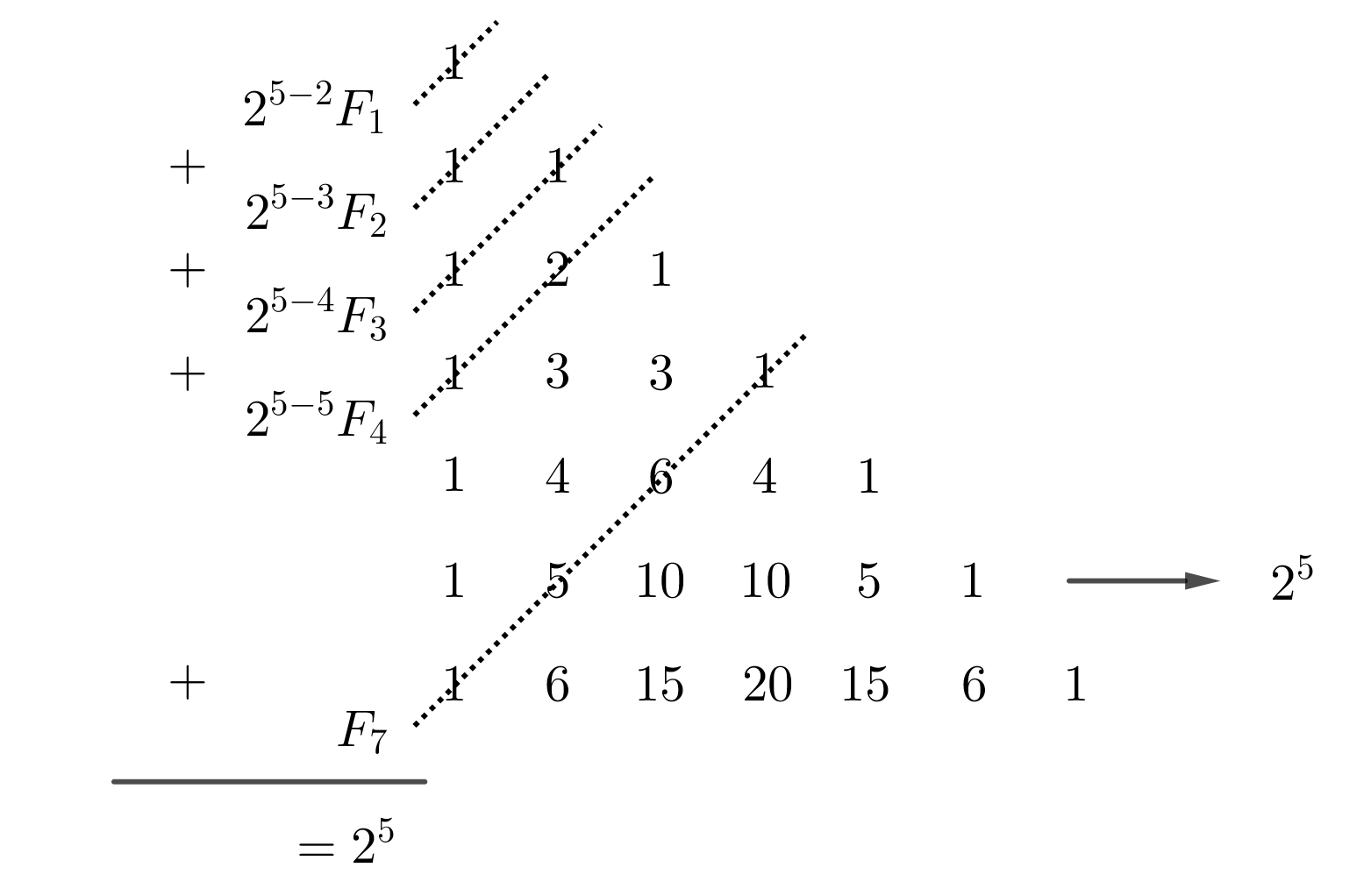}
\begin{caption}{Vizualization of Corollary \ref{SommeMignonne} for $n=5$.}\end{caption}
\label{Pascal}
\end{figure}

\section{Some generalizations and questions}\label{Gene}

\subsection{Modifying the Fibonacci moves}\label{Modif}

Here, we briefly consider alternative ways of defining the allowed moves, extending in a natural way the Fibonacci moves. We write $\Delta_n^{n'}$ for the set of disks $d_k$ with $n'\leqslant k\leqslant n$ (so $\Delta_n^{n'}=\Delta_n^1=\Delta_n$ for $n'\leqslant 1$ and $\Delta_n^{n'}=\varnothing$ for $n'>n$).

\begin{definition} Let $p\geqslant 1$ and $q\geqslant 0$ be two integers. Let $X$ and $Y$ be two different pegs of some state such that, for some $k\in\Delta_n$, we have $X=\Delta_k^{k-p+1}X'$ and $Y=\Delta_{k-p}Y'$. Write $Z$ for the third peg of the state. We define a {\em $(p,q)$-move} as a move that consists in putting simultaneously all the disks of $\Delta_k^{k-p+1-q}$ onto $Z$, i.e.:
\[\Delta_k^{k-p+1}X'\sqcup \Delta_{k-p}Y'\sqcup Z\longrightarrow X'\sqcup \Delta_{k-p-q}Y'\sqcup \Delta_k^{k-p+1-q} Z.\]
We will talk about the {\em $(p,q)$-Tower of Hanoi} for the Tower of Hanoi puzzle in which only $(p,q)$-moves are allowed.
\end{definition}

Note that the $(1,0)$ case is the classical puzzle, and that the $(1,1)$ one is the Tower of Hanoi-Fibonacci puzzle.

\begin{theorem}\label{Optimalpq} The $(p,q)$-Tower of Hanoi puzzle with $n$ disks admits a solution for any $n\geqslant 0$. There exists only one optimal algorithm for it, that needs exactly $m_n$ $(p,q)$-moves, where the sequence $(m_n)_{n\in\Z}$ is defined by 
\[m_n=\left\{\begin{array}{cl}0&\mbox{for $n\leqslant 0$}\\ m_{n-p}+m_{n-p-q}+1&\mbox{for $n>0$.}\end{array}\right.\]
\end{theorem}

\begin{proof} For $n\leq p$, the $(p,q)$-move $(\Delta_n,\varnothing,\varnothing)\longrightarrow(\varnothing,\varnothing,\Delta_n)$ is allowed, so we have $m_n=1$ for any $n\leqslant p$, which correspond to the formula stated in the theorem.

For $n>p$, the optimal solution is provided by the sequence of critical states, each of which needing, by induction, the number of moves written on its arrow:
\[(\Delta_n,\varnothing,\varnothing)\stackrel{m_{n-p}}{\longrightarrow}(\Delta_n^{n-p+1},\Delta_{n-p},\varnothing)\stackrel{1}{\longrightarrow}(\varnothing,\Delta_{n-p-q},\Delta_n^{n-p+1-q})\stackrel{m_{n-p-q}}{\longrightarrow}(\varnothing,\varnothing,\Delta_n).\]
By summing the moves we get the expected formula. 
\end{proof}

There is no serious doubt that generalization of Zeckendorf-Fibonacci, Gray-like codes and pseudo-edges of the graph $\mathcal{H}_n$ can be given for $(p,q)$-moves, but some additional technicalities may be quite hard to overcome. For example, the case $p=q=2$ provides a sequence $m_n$ which is not strictly increasing (since $m_{2n-1}=m_{2n}$), hence a convenient numeration system derived from it is probably not as simple as the Zeckendorf one for the Fibonacci sequence of the case $p=q=1$. The study of the corresponding graph may be a little bit tricky as well to be extended to $(p,q)$-moves.

\subsection{Restricting the moves between pegs}\label{Tribo}

Possible variants on the classical puzzle consist in allowing moves only between some pegs. For example, in the {\em clockwise-cyclic} variant introduced in \cite{Atkinson1981}, additionally to the classical rules of Section \ref{HanoiClassique}, a disk can move only from $A$ to $B$, from $B$ to $C$ or from $C$ to $A$.

Any variant of this type can be defined by an oriented graph with set of vertices $\{A,B,C\}$, an arrow $XY$ standing for the fact that moves from the peg $X$ to the peg $Y$ are allowed. The sensible variants of this kind (i.e. for which the puzzle is solvable for any $n$) are the ones for which the corresponding graph is strongly connected \cite[Theorem 8.4]{Hinz2019}. We will not consider all possible cases here, but only mention briefly the  {\em linear} variant, in which the allowed moves are those from $A$ to $B$, from $B$ to $A$, from $B$ to $C$ and from $C$ to $B$. It is well-known that, for such a restriction, the optimal algorithm for the classical puzzle needs $3^n-1$ moves, so, since the number of distincts states is $3^n$, the linear puzzle also provides the ``worst solution'', that is: the longest solution that does not come back to any state already met.

Now, consider the linear variant for the Tower of Hanoi-Fibonacci, in which a $k$-Fibonacci move is allowed iff it makes $d_k$ going from $A$ to $B$, from $B$ to $A$, from $B$ to $C$ or from $C$ to $B$. Write again $m_n$ for the minimal number of moves to solve ths variant with $n$ disks. The optimal solution is then given by the following recursive description (for $n\geqslant 3$):
\[(\Delta_n,\varnothing, \varnothing)\stackrel{m_{n-1}}{\longrightarrow}(n,\varnothing,\Delta_{n-1})\stackrel{1}{\longrightarrow}(\varnothing,\Delta_n^{n-1},\Delta_{n-2})\stackrel{1}{\longrightarrow}\]
\[(\Delta_{n-1}^{n-2},n,\Delta_{n-3})\stackrel{m_{n-3}}{\longrightarrow}(\Delta_{n-1},n,\varnothing)\stackrel{1}{\longrightarrow}(\Delta_{n-2},\varnothing,\Delta_n^{n-1})\stackrel{m_{n-2}}{\longrightarrow}(\varnothing,\varnothing,\Delta_n).\]

Hence, the sequence $(m_n)_n$ is given by $m_0=0$, $m_1=2$, $m_2=5$ and, for any $n\geqslant 3$, $m_n=m_{n-1}+m_{n-2}+m_{n-3}+3$ (a kind of a {\em Tribonacci} sequence).

As regards the other variants derived from the restriction of moves between pegs, there is probably no specific difficulty to deal with them in the context of Fibonacci moves (or $(p,q)$-moves), apart from the fact that some of these variants already involve linear recurring sequences of order $6$ in the classical Tower of Hanoi, so are possibly tiresome to describe in our even more technical context.

More interesting would be to find a general way to derive the sequence of moves (or at least of number of moves) from the conjunction of the two kinds of rules. For example, is it possible to deduce the previous Tribonacci sequence directly from what we separately know from the linear variant of the classical puzzle and from our study of the Tower of Hanoi-Fibonacci, instead of the recursive description we presented?

\subsection{Further questions}\label{Questions}

We could also consider even more general rules for the moves. For example, we could allow moves of the form $kX'\sqcup\Delta_{k-1}Y'\sqcup Z\longrightarrow X'\sqcup\Delta_{n-3}(n-1)Y'\sqcup(n-2)nZ$, and so on. One may wonder if two different rules can lead to the same sequence, hence asking for the links between these rules.

Eventually, a deeper work would be to obtain a theoretical way to find from a linear recurring sequence some natural rules for the Tower of Hanoi for which the number of moves of the optimal algorithm would be given by the sequence. This will probably involve a more precise definition of a ``natural rule''. (For example, we may ask whether we can always restricts the study to markovian moves, i.e. moves for which their legality depends only on the initial and final states.) In a sense, answering this question would truly complete Lucas' original assertion.

\medskip

\noindent MSC2020: 00A08, 05C20, 11B39, 68R15, 94B25.

\end{document}